\documentclass[conference]{IEEEtran}
\IEEEoverridecommandlockouts
\usepackage{cite}
\usepackage{amsmath,amssymb,amsfonts} 
\usepackage{bm}                       
\usepackage{algorithmic}
\usepackage{graphicx}
\usepackage{textcomp}
\usepackage{xcolor}                    
\usepackage{booktabs} 
\usepackage{floatrow}
\usepackage{subfig}
\usepackage[font=small]{caption} 
\usepackage{framed}
\usepackage{pifont}
\let\labelindent\relax
\usepackage{enumerate}
\usepackage{enumitem}
\usepackage[colorlinks=true, linkcolor=blue, citecolor=blue, urlcolor=blue]{hyperref}
\usepackage[capitalize]{cleveref}
\newtheorem{theorem}{Theorem}

\def\BibTeX{{\rm B\kern-.05em{\sc i\kern-.025em b}\kern-.08em
    T\kern-.1667em\lower.7ex\hbox{E}\kern-.125emX}}

\begin{document}

%\title{UAV-enabled Victim-Sound-Detection-and-Location (V$^2$SDLS): System Design and Experiments}

\title{Sky-Ear: An Unmanned Aerial Vehicle-Enabled Victim Sound Detection and Localization System \vspace*{-0.6cm} }

\author{
    \IEEEauthorblockN{
        Yi Hong\IEEEauthorrefmark{1}\IEEEauthorrefmark{2},
        Mingyang Wang\IEEEauthorrefmark{1}\IEEEauthorrefmark{2}, 
        Yalin Liu\IEEEauthorrefmark{1}\IEEEauthorrefmark{3}, 
        Yaru Fu\IEEEauthorrefmark{1}, 
        Kevin Hung\IEEEauthorrefmark{1}, 
        Bishenghui Tao\IEEEauthorrefmark{1}
    }
    \IEEEauthorblockA{
        \IEEEauthorrefmark{1}School of Science and Technology, Hong Kong Metropolitan University, Hong Kong 999077, China\\
        \IEEEauthorrefmark{2}Equal contribution, \IEEEauthorrefmark{3}Corresponding author.\\
        E-mails: yhonghk@gmail.com, wangmingyang079@outlook.com, 
        \{ylliu, yfu, khung, btao\}@hkmu.edu.hk
        \vspace*{-1.1cm}
    }
    \thanks{This work was supported in part by the Research Grants Council of the Hong Kong Special Administrative Region, China, under Project UGC/FDS16/E15/24, in part by the Hong Kong Metropolitan University Research Grant under Project PFDS/2025/34, and in part by the UGC Research Matching Grant Scheme under Project 2024/3003 and Gekko Lab.}
}

\maketitle

\begin{abstract}
Unmanned Aerial Vehicles (UAVs) are increasingly deployed in search-and-rescue (SAR) missions, yet continuous and reliable victim detection and localization remain challenging due to on-board hardware constraints. This paper designs an UAV-Enabled Victim Sound Detection and Localization System (called ``Sky-Ear'' for brevity) to achieve energy-efficient acoustic sensing and sound detection for SAR. Sky-Ear enables the ``ear'' of the UAV with a circular-shaped microphone array, and the array conducts continuous audio recordings during the UAV's flight. In Sky-Ear, a two-stage (Sentinel and Responder) audio processing method is developed for energy-consuming and highly reliable sound detection. In the Sentinel stage, a Masking autoencoder (MAE)-based sound detection mechanism is designed to analyze frequency-time acoustic features. For improved precision, a continuous localization method is designed by optimizing detected directions from multiple observations. Extensive simulation experiments are conducted to validate the system's performance in terms of victim detection accuracy and localization error. 
\end{abstract}

\begin{IEEEkeywords}
Unmanned Aerial Vehicles, Search and Rescue, Acoustic Sensing, Anomaly Detection, Continuous Localization
\end{IEEEkeywords}

\vspace{-0.2cm}
\section{Introduction}
\vspace{-0.1cm}
With high maneuverability, emerging unmanned aerial vehicles (UAVs) offer a promising solution for rapid victim search and rescue (SAR) in complex terrains such as dense forests and deserts~\cite{lyu2023unmanned}. Mainstream UAVs primarily rely on visual sensing systems, i.e., applying cameras to capture images or videos, thus enabling high-fidelity object detection and tracking for SAR~\cite{app14020766}. However, these visual systems impose a heavy payload burden on hardware-constrained UAVs and are vulnerable to line-of-sight (LoS) obstructions like forest canopies or thick fog in SAR scenarios~\cite{app14020766}. Thermal infrared imaging can reduce payload weight but works only at night and loses reliability in high-heat backgrounds or when survivors are hidden under debris or thick clothing~\cite{katkuri2024autonomous}. In contrast, acoustic sensing systems, i.e., recording the audio signals, provide reliable spatio-temporal information even when visual or thermal signals are severely distorted. Meanwhile, a lightweight on-board payload is enough to process audio signals, making the acoustics sensing system a sustainable, reliable, and cost-effective solution for UAV-enabled SAR missions~\cite{fraternali2025enhancing}.

In acoustic sensing, microphone array-based audio processing and Time Difference of Arrival (TDoA) estimation have laid the theoretical foundation for sound source localization, and thus are widely applied in robot audition, surveillance, and human-computer interaction systems~\cite{mikesikowska2024classification}. Importantly, these technologies require always-on multi-channel microphone arrays running complex beamforming algorithms, which introduce inevitable computational overhead and energy consumption~\cite{s19194326}. However, in UAV-enabled SAR missions, during flights lasting several hours, the blind search phase (lacking effective victim signals) often occupies more than 90\% of the total mission time~\cite{11009148}. Maintaining a full-power microphone array and high-frequency spatial filtering during these prolonged silent periods is redundant, directly damaging the energy efficiency. Conversely, while low-power periodic sleep architectures have been widely explored, they often result in an intolerable miss-detection rate for audio streams~\cite{deAlcantaraAndrade2019}.

\begin{figure*}[t]
    \centering
    \includegraphics[width=16.2cm]{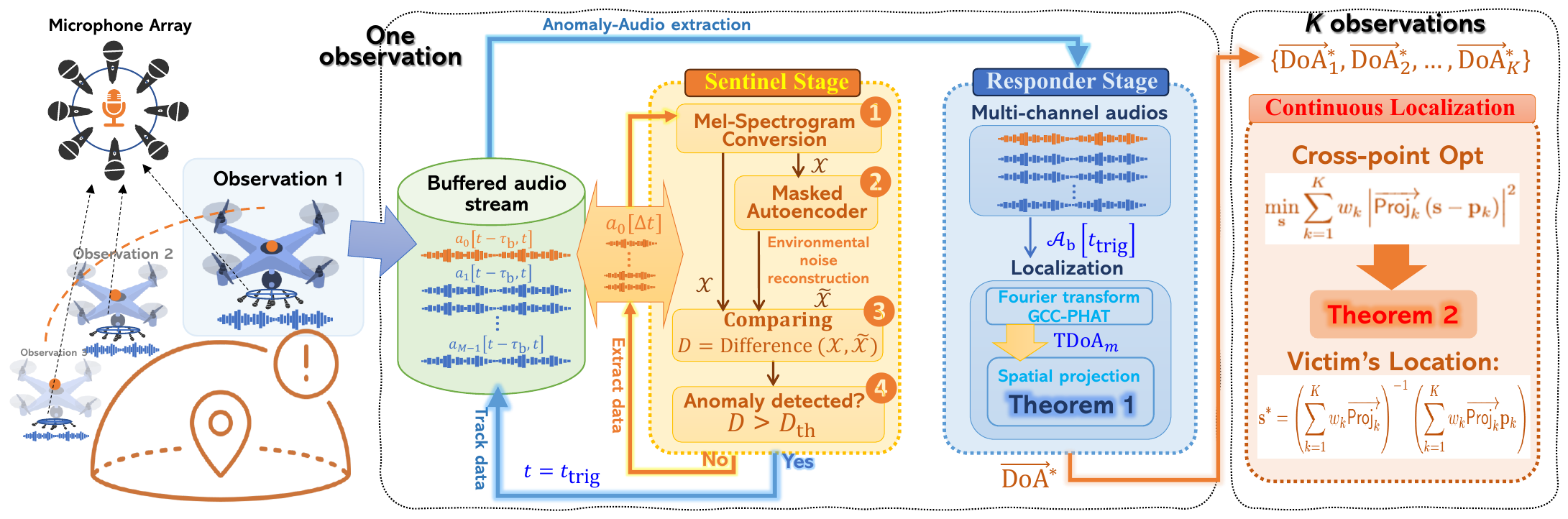}
    \caption{An Unmanned Aerial Vehicle-Enabled Victim Sound Detection and Localization System (Sky-Ear).}
    \label{fig:system_architecture}
     \vspace{-0.4cm}
\end{figure*}

To mitigate the computational burden of multi-channel audio processing while maintaining effective victim sound detection, we propose an energy-efficient two-stage audio processing method. Unlike existing sound event detection methods that rely on continuously turned on microphone arrays, Sky-Ear divides acoustic inspection and localization into two stages. In our method, the single-channel-based Sentinel stage is responsible for continuous listening and anomaly detection. Once an anomaly is detected, i.e., the suspected victim's sound, the multi-channel-based Responder stage is triggered for accurate localization. Along the UAV's flight trajectory, multiple two-stage observations are performed for continuous localization with improved precision. The main contributions of this paper are summarized below:
\vspace{-0.1cm}

\begin{itemize}[leftmargin=*,align=left]
%\small
    \item An Unmanned Aerial Vehicle-Enabled Victim Sound Detection and Localization System (Sky-Ear) is designed to achieve efficient acoustic sensing and sound detection for SAR. Based on a circular-shaped microphone array, two-stage (Sentinel and Responder) audio processing is developed for energy-efficient and reliable sound detection. For improved precision, a continuous localization is used by optimizing detected directions from multiple observations. 
    \item A Masking autoencoder (MAE)-based sound detection method is designed in the Sentinel stage. MAE processes the audio's Mel-spectrogram to accurately analyze frequency-time acoustic features. The suspected victim's sound is detected as an anomaly by comparing the recorded audio with a learned noise from MAE. Multiple MAE models are trained to adapt to varying noise for different SAR scenarios. 
    \item Extensive simulation experiments are conducted to validate the system's performance in terms of victim detection and localization. The comprehensive audio datasets of environmental background noises, drone propeller noises, and victim sounds are created to fine-tune and verify the performance of MAE in different SAR scenarios. The effectiveness of the continuous localization is validated by obtaining the improved localization via multiple observations.
\end{itemize}
\vspace{-0.2cm}

\section{System Design}
\vspace{-0.1cm}
%As shown in \cref{fig:system_architecture}, %we design a UA\underline{V}-enabled \underline{V}ictim-\underline{S}ound-\underline{D}etection-and-\underline{L}ocalization \underline{S}ystem (V$^2$SDLS). In 
\cref{fig:system_architecture} shows the design of ``Sky-Ear'', where a UAV flies over a search region to perform a SAR mission, i.e., searching and localizing the victim. The UAV mounts an $M$-element circular microphone array, denoted by $\mathcal{M}=\{0,1,...,M-1\}$, where the microphone $0$ is located in the array center and $\mathcal{M}\setminus 0$ are uniformly distributed as a circular array around the center. %The central microphone $0$ and the surrounding circular array $\mathcal{M}\setminus 0$ are installed at different heights and therefore have different 3-dimensional(3D) coordinates. 
During the UAV's flight, the microphone array continuously records an $M$-channel audio stream, denoted by $\mathcal{A}=\{a_{m},m\in \mathcal{M}\}$. The collected audio $\mathcal{A}$ is then used to detect the victim's sound and localize the victim. Particularly, the system performs a continuous two-stage procedure:
%\vspace{-0.1cm}

{\small
%\begin{framed}
\begin{enumerate}[leftmargin=*,align=left]
    \item At a hovering point, the UAV performs one observation, i.e., collecting audio streams $\mathcal{A}$, and executes a two-stage processing:
    \begin{enumerate}[leftmargin=*,align=left]
        \item Sentinel Stage, using $a_{0}$ for ``listening and detection'', and 
        \item Responder stage, using $\mathcal{A}$ for ``full-channel localization''.
    \end{enumerate}
    \item During the flight, the UAV hovers at multiple points and performs multiple observations, thus obtaining multiple localized directions of the victim. Summarizing multiple observed results, the system achieves continuous localization with improved precision.
\end{enumerate}
%\end{framed}
}
Due to the UAV's storage constraint, %the recorded audio streams need to be continuously updated. Therefore, 
a ring buffer mechanism is considered to record the continuous audio streams, i.e., updating $\mathcal{A}$ in the ring buffer with the length $\tau_{\mathrm{b}}$. Let $\mathcal{A}_{\mathrm{b}}$ be the buffered audio stream at the time $t$, which is given by
\vspace{-0.2cm}

{\small 
\begin{equation*}
    \mathcal{A}_{\mathrm{b}}=\{a_{m}[t-\tau_{\mathrm{b}},t],m\in \mathcal{M}\}.
\end{equation*}
}%
%where $t$ is the time of an observation. 
The audio for two-stage processing of one observation is extracted from $\mathcal{A}_{\mathrm{b}}$. Overall, the two-stage procedure, the ring buffer mechanism, and the continuous localization ensure low energy consumption while achieving high-precision localization of the victim. Next, we give the detailed design. 
%During Sentinel stage, an abnormal-sound detection algorithm - SpecMAE is designed and used to continuously process the sound-flow and adaptively detect abnormal sounds (the suspected victim's sound, e.g., screaming or shouting for help) from the normal environmental sounds (e.g., UAV propeller noise and environmental background noise). Once an abnormal signal is detected, the system is triggered to perform Responder stage for high-precision localization. 
%Throughout the SAR mission cycle, the full array $\mathcal{M}$ remains in a continuous recording state.
\vspace*{-0.2cm}
\subsection{Sentinel Stage: MAE-based Anomaly Detection}
\label{subsec: sentinel_stage}
\vspace*{-0.1cm}
Given the audio stream $a_{0}[t-\tau_{\mathrm{b}},t]$, the Sentinel stage is to detect the anomaly-audio sequence, i.e., where the suspected victim's sound occurs, e.g., screaming or shouting. Masked autoencoder (MAE) is used to analyze the accurate acoustic feature in Mel-spectro and detect the anomaly by continuously browsing a group of audio clips $\{a_{0}[\Delta t]\}$ from $a_{0}[t-\tau_{\mathrm{b}},t]$, where $[\Delta t]\subset [t-\tau_{\mathrm{b}},t]$ is the audio length to be processed. %Particularly, MAD continuously processes a short audio clip $a_{0}[\Delta t]$ from $a_{0}[t-\tau_{\mathrm{b}},t]$, where $\Delta t\subset  [t-\tau_{\mathrm{b}},t]$ is the audio clip's period and $i=\{1,2,...\}$ is the clip index. %To adaptively learn different sound features, the clip $a_{0}[\Delta t]$ is converted to a 2-dimensional (2D) log-Mel spectrogram image (with the frequency and time dimensions).
The MAE-based anomaly detection includes 4 steps:

{\small
\begin{enumerate}[leftmargin=*,align=left]
    \item[\ding{172}] The coming audio clip $a_{0}[\Delta t]$ is first converted to a two-dimensional (2D) Mel-spectrogram image, denoted by $\mathcal{X}\in \mathbb{R}^{F\times T}$, where $F\times T$ indicates frequency and time dimensionality. % $\mathcal{X}$ captures the frequency-time dimensional sound features.
    \item[\ding{173}] A MAE model is used to mask partial info from $\mathcal{X}$ and reconstruct a Mel-spectrogram image $\tilde{\mathcal{X}}$ with accurate sound features from environments, e.g., drone propeller noise and background noise. 
    \item[\ding{174}] The difference $D_{\mathrm{re}}$ between $\mathcal{X}$ and $\tilde{\mathcal{X}}$ is calculated to detect an anomaly. The value of $D_{\mathrm{re}}$ is obtained from a Top$-K$ scoring strategy to reduce the impact of large-energy background noises. 
    \item[\ding{175}] Let $D_{\mathrm{th}}$ be a threshold for detecting an anomaly. If and only if $D_{\mathrm{re}} > D_{\mathrm{th}}$, an anomaly is detected and triggers the Responder stage for the subsequent victim's localization. 
\end{enumerate}
}%
The efficacy of the above procedure lies in the accurate noise audio reconstruction capability of the used MAE model. %To achieve this, we pre-train and fine-tune multiple MAE models to adapt to different environments, e.g., desert and forest. 
The detailed workflow for anomaly detection, pre-training, and fine-tuning of the MAE model is given in \cref{sec:specmae,sec:experiments}.

\begin{figure*}
\begin{framed}
\vspace{-0.3cm}
\begin{theorem} 
Given $\mathcal{A}_{\mathrm{b}}[t_{\mathrm{trig}}]$ in one observation, $\overrightarrow{\mathsf{DoA}}^{*}$ is calculated by
\vspace{-0.3cm}

{\footnotesize
    \begin{align*}
    &\overrightarrow{\mathsf{DoA}}^{*}
    =(\mathbf{G}^T \mathbf{G})^{-1} \mathbf{G}^T \mathbf{V},\text{ where }
    \mathbf{G} = [G_1; G_2; \dots; G_{M-1}], \mathbf{V} 
    = \left [ V_{1},...,V_{M-1} \right ]^T, 
    \forall m \in \mathcal{M}\setminus 0 : G_{m}=(\mathbf{r}_{m} - \mathbf{r}_0)^T,\\
    &V_{m}=\mathsf{TDoA}_{m} \cdot v_{\mathrm{s}},a_{m}'=a_{m}[ t_{\mathrm{trig}}-\tau_{\mathrm{w}},t_{\mathrm{trig}}],
    \mathsf{TDoA}_m 
    = \arg\max_{\tau} \int_{-\infty}^{\infty} \frac{\mathcal{F}(a_{m}') \mathcal{F}^*(a_{0}')}{|\mathcal{F}(a_{m}') \mathcal{F}^*(a_{0}')|} e^{j2\pi f\tau} \mathrm{d}f,
    \end{align*}
}%
\vspace{-0.4cm}

\noindent $a_{0}'=a_{0}[ t_{\mathrm{trig}}-\tau_{\mathrm{w}},t_{\mathrm{trig}}]$, $\mathcal{F}\{\cdot\}$ is the Fourier transform, $\mathcal{F}^*\{\cdot\}$ is its complex conjugate, and $v_{\mathrm{s}}$ is the speed of sound. %For $m \in \mathcal{M}\setminus 0 $, $G_{m}=(\mathbf{r}_{m} - \mathbf{r}_0)^T$, $V_{m}=\mathsf{TDoA}_{m} \cdot v_{\mathrm{s}}$, and $a_{m}'=a_{m}[ t_{\mathrm{trig}}-\tau_{\mathrm{w}},t_{\mathrm{trig}}]$.
\label{the1}
\end{theorem}
%\vspace{-0.5cm}
\noindent \textit{Proof.} The proof is provided in \cite[Appendix~A]{complementary}. \hfill $\blacksquare$

\begin{theorem} 
    Based on $K$ observations, $\mathbf{s}^{*}$ is calculated by $\mathbf{s}^{*} = \left ( \sum_{k=1}^K  \overrightarrow{\mathsf{Proj}_k}\right )^{-1} \left ( \sum_{k=1}^K  \overrightarrow{\mathsf{Proj}_k}\mathbf{p}_k \right )$, where \\$\overrightarrow{\mathsf{Proj}_k}=\mathbf{I} - \overrightarrow{\mathsf{DoA}}_{k}^{*} (\overrightarrow{\mathsf{DoA}}_{k}^{*})^T$, $\mathbf{I}$ is the $3 \times 3$ identity matrix and $\mathcal{F}^{-1}\{\cdot\}$ denotes the inverse Fourier transform.
\label{the2}
\end{theorem}
%\vspace{-0.3cm}
\noindent \textit{Proof.} The proof is provided in \cite[Appendix~B]{complementary}. \hfill $\blacksquare$
\vspace{-0.25cm}
\end{framed}
\vspace*{-0.2cm}
\end{figure*}

\vspace*{-0.1cm}
\subsection{Ring Buffer: Multi-Channel Audio Tracking}
\vspace*{-0.1cm}
%To ensure that the Responder can access the complete anomalous audio signal for high-precision localization after an anomaly is triggered, 
%the system has introduced a Retroactive Data Bridge, 
%a historical data tracking mechanism. 
% During the Sentinel listening stage, although the abnormal audio signal detected by the responder array $\mathcal{M}_{\mathrm{aux}}$ is not used for computation, the continuous sampling $y_i(t)$ from any microphone $i \in \mathcal{M}_{\mathrm{aux}}$ is continuously pushed into a ring buffer $\mathcal{B}$ of length $\tau_{\mathrm{b}}$ in real-time, expressed as
% \begin{equation}
% \mathcal{B}_i(t) = \left\{ y_i(\tau) \mid \tau \in [t - \tau_{\mathrm{b}}, t] \right\}.
% \label{eq:buffer}
% \end{equation}
Upon detecting an anomaly during the Sentinel stage, the system is triggered to perform the Responder stage. Let $t_{\mathrm{trig}}$ denote the exact timestamp of the triggering event. %Then the system immediately performs historical extraction of the anomaly-audio stream from the $M-1$ microphones $\mathcal{M}\setminus 0$. %the buffered audio stream of $\mathcal{M}\setminus 0$, i.e., $\{a_{m}[T_{\mathrm{b}}],T_{\mathrm{b}}\subset [t-\tau_{\mathrm{b}},t],t=t_{\mathrm{trig}},m\in \mathcal{M}\setminus 0\}$. %To prepare the processed audios for Responder stage, the system retroactively extracted multi-channel time-domain data frame $S_{k}[T_{\mathrm{win}}]$ from the initial continuous data stream is defined as
To prepare the processed audio sequences for the Responder stage, the system needs to track the historical anomaly-audio stream collected from the $M$ microphones. %, i.e., from the audio buffer $\mathcal{A}_{\mathrm{b}}$. 
Let $\mathcal{A}_{\mathrm{b}}[t_{\mathrm{trig}}]$ be the extracted $M$ channel audio sequences, which is given by
\vspace{-0.3cm}

{\small 
\begin{align}
&\mathcal{A}_{\mathrm{b}}[t_{\mathrm{trig}}] = \{a_m[t_{\mathrm{trig}}-\tau_{\mathrm{w}},t_{\mathrm{trig}}],m\in \mathcal{M}\}, \label{eq:truncation_t}
\end{align}
}%
where $\tau_{\mathrm{w}}\leq \tau_{\mathrm{b}}$ is the extracted sequence length, this ensures that the extracted sequence is securely retained within the maximum reserve capacity $\tau_{\mathrm{b}}$ of the circular buffer. To ensure the statistical completeness of $\mathcal{A}_{\mathrm{b}}[t_{\mathrm{trig}}]$ for Responder stage, we define $\tau_{\mathrm{w}}$ as the valid retroactive period given by $\tau_{\mathrm{w}} = \tau_{\mathrm{retro}} + \tau_{\mathrm{post}}$, where the retroactive window $\tau_{\mathrm{retro}}$ captures the anomaly's onset prior to detection, and the post-observation window $\tau_{\mathrm{post}}$ provides sufficient data resolution for cross-correlation.

\vspace*{-0.2cm}
\subsection{Responder Stage: Multi-Channel Localization}
\label{subsec:responder}
\vspace*{-0.1cm}
%During the UAV's flight, a specific observation point (i.e., hovering at a position) can be triggered to perform the instant localization of the victim(s) based on $\mathcal{A}_{\mathrm{b}}[t_{\mathrm{trig}}]$. 
Based on the extracted $M$ audio sequences, i.e., $\mathcal{A}_{\mathrm{b}}[t_{\mathrm{trig}}]$ that contains the detected anomaly-audio signal, the Responder stage makes instant localization of the victim. Particularly, the Direction of Arrival (DoA) of the suspected victim's sound at the UAV can be calculated from $\mathcal{A}_{\mathrm{b}}[t_{\mathrm{trig}}]$. The DoA of the victim's sound is denoted by the 3D unit vector $\overrightarrow{\mathsf{DoA}} = [\mathsf{DoA}_x, \mathsf{DoA}_y, \mathsf{DoA}_z]^T$. Let $\mathbf{r}_m= [\mathsf{r}_x, \mathsf{r}_y, \mathsf{r}_z]^T$ ($\forall m \in \mathcal{M}$) be the known 3D coordinate vectors of the $m$-th microphone. Let $\mathsf{TDoA}_m$ ($\forall m \in \mathcal{M}$) is TDoA of the anomaly signal at the microphone array $\mathcal{M}\setminus 0$ compared with the central microphone. %Based on the far-field plane wave assumption, acoustic waves emitted by the victim and arriving at the microphone array are approximated as parallel wavefronts. Therefore, t
The time difference $\mathsf{TDoA}_m$ is theoretically modeled as the spatial projection of the unknown victim direction $\overrightarrow{\mathsf{DoA}}$ onto the rays to the microphones. Thus, the optimal $\overrightarrow{\mathsf{DoA}}^{*}$ is derived in \cref{the1}.

\vspace{-0.25cm}
\subsection{UAV-based Continuous Localization}
\vspace*{-0.1cm}
%Let the spatial location of the victim be $s = [x_s, y_s, z_s]^T$. 
The calculated $\overrightarrow{\mathsf{DoA}}^{*}$ in \cref{the1} estimates the direction of the victim at a single observation, i.e., the UAV hovers at a position. To make a precise localization of the victim, we design a continuous Localization Mechanism as follows:
\begin{enumerate}[leftmargin=*,align=left]
    \item Multiple observations along the UAV's flight trajectory are considered to collect the victim's directions. 
    \item The victim's localization is determined by using a globally optimized cross-point of multiple observed directions. 
\end{enumerate}
%At the $k$-th distinct spatial position $\mathbf{p}_k$, the system captures a corresponding direction vector $\hat{\mathbf{u}}_k$ with a confidence weight $w_k$.Ideally, these $K$ spatial rays intersect exactly at the target point $s$. However, due to inevitable measurement noise, they typically form non-intersecting skew lines in 3D space. 
%To resolve this robustly, we transform the multi-view localization into an optimization problem: finding an optimal point $\mathbf{s}$ that minimizes the sum of the weighted squared orthogonal distances to all $K$ rays. 
Given $K$ observations, let $\mathbf{p}_k$ and $\overrightarrow{\mathsf{DoA}}_{k}^{*}$ denote the UAV's 3D coordinates (known along the trajectory) and the calculated DoA (referring to \cref{the1}) at the observation $k\in \{1,2,...,K\}$. Let $\mathbf{s} = [x_s, y_s, z_s]^T$ be the 3D coordinates of the victim. By summarizing multiple observed directions, the globally optimized cross-point $\mathbf{s}^{*}$ is derived in \cref{the2}.

\vspace{-0.15cm}
\section{MAE For Anomaly Detection} \label{sec:specmae}
\vspace*{-0.2cm}
%This section introduces the MAE model, the MAE-based anomaly detection, and the pretraining of MAE models.
\subsection{The MAE Model}
\vspace*{-0.1cm}
\subsubsection{Mel-spectrogram Discretization and Masking}
%\vspace*{-0.1cm}
To prepare a Mel-spectrogram image $\mathcal{X}\in \mathbb{R}^{F\times T}$ for the MAE, we discretize $\mathcal{X}$ into $N$ uniform image patches (with the patch size $P\times P$) along its frequency-time dimensions. Then, the discretized $\mathcal{X}$ is deemed as a sequence of patches, denoted by
\vspace{-0.8cm}

{\small 
\begin{align*}
    \mathbf{X} = \{X_n, n\in \{1,2,...,N\} \},
\end{align*}
}%
\vspace{-0.5cm}

\noindent where $X_n\in \mathbb{R}^{P \times P}$ and $N=(F/P)\times (T/P)$. For $n\in\{1,2,...,N\}$, with the frequency-time domain acoustic features in a discretized Mel-Spectrogram image, $X_n$ will be compared with its background noise to detect the anomaly. The background noise is reconstructed from a masked version of $X_n$. Particularly, a masked patch sequence $\ddot{\mathbf{X}}$ is generated by randomly hiding the info of $N\time \rho$ patches with $\rho$ being a masking ratio, %. %To reduce computational overhead, the hidden patches for the index set $\mathcal {U}$ are explicitly discarded, only performs operations on the unmasked patches, thus, $\ddot{\mathbf{X}}$ is given by
% \vspace{-0.3cm}
% {\small 
% \begin{align*}
%     \ddot{\mathbf{X}} = \{X_i,i\in \mathcal{V} \},
% \end{align*}
% }%
which is given by $\ddot{\mathbf{X}} = \{X_i,i\in \mathcal{V} \}$, where $\mathbf{0}\in \mathbb{R}^{P \times P}$, $|\mathcal{U}|=N\rho$, and $\mathcal{V} + \mathcal{U} = \{1,2,...,N\}$.

% \begin{figure}[t]
%     \centering
%     \includegraphics[width=7.5cm]{Images/SpecMAE _Architecture.png}
%     \caption{The architecture of the SpecMAE.}
%     \label{fig:SpecMAE _Architecture}
% \end{figure}

% Subsequently, we re-concatenate the classification patch to the head of the sequence and add the 2D positional encoding $\mathbf{P}_{\mathrm{dec}}$, yielding the final input $\mathbf{Z}_{\mathrm{dec}}$ for the decoder, given by
% \begin{equation}
% \mathbf{Z}_{\mathrm{dec}} = [\mathbf{z}_{\mathrm{cls}}; \mathbf{h}_1; \mathbf{h}_2; \dots; \mathbf{h}_N] + \mathbf{P}_{\mathrm{dec}}.
% \label{eq:decoder_in}
% \end{equation}
% \subsubsection{Encoder}
%The encoder $f_ {\mathrm{enc}} (\cdot)$ aims to map the initial patch $\mathbf{X}$ block to a high-dimensional matrix representation. 

\subsubsection{Encoder}
%the encoder $f_{\mathrm{enc}}(\cdot)$ only processes the patches in the visible set $\mathcal{V}$ to , map the initial patch $\mathbf{X}$ block to a high-dimensional matrix representation, extracting deep acoustic features. For each visible patch $X_v$ ($v \in \mathcal{V}$), the system first maps it into a $D$-dimensional latent space through a linear embedding matrix $\mathbf{E} \in \mathbb{R}^{P^2 \times D}$, superimposes the encoder-specific positional encoding $\mathbf{P}_{\mathrm{enc}}$, and prepends a learnable classification patch $\mathbf{z}_{\mathrm{cls}} \in \mathbb{R}^D$ to the sequence. After being processed by multiple Transformer blocks, the encoder outputs the deep visible feature sequence $\mathbf{Z}_{\mathrm{enc}}$, which is given by
The encoder (a Transformer) learns the deep acoustic feature from the unmasked patches $\{X_i,i\in \mathcal{V} \}$ by outputting a 1D feature sequence $\mathbf{Z}_{\mathrm{enc}}$ as follows
\vspace{-0.3cm}

{\small
\begin{align}
    &\mathbf{Z}_{\mathrm{enc}} = f_{\mathrm{enc}}\left( \left[ \mathbf{z}_{\mathrm{cls}}; \mathbf{Z}_{v_1}; \mathbf{Z}_{v_2}; \dots; \mathbf{Z}_{v_{|\mathcal{V}|}} \right] \right),\notag \\ 
    &\text{where }\forall  v \in \{v_1,v_2,...,v_{|\mathcal{V}|}\}:\mathbf{Z}_{v}  = \mathbf{x}_{v}\mathbf{E} +  \mathbf{p}^{\mathrm{enc}}_{v}. \label{eq:encoder_op}
\end{align}
}%
In \cref{eq:encoder_op}, $f_{\mathrm{enc}}(\cdot)$ is the Transformer block, $\mathbf{z}_{\mathrm{cls}} \in \mathbb{R}^{1\times D}$ is a classification label added before the sequence to concatenate these individual feature vectors to form a one-dimensional input sequence. For $v\in\{v_1,v_2,...,v_{|\mathcal{V}|}\}$, $\mathbf{Z}_{v}$ characterizes the deep feature of the patch $X_{v}$, $\mathbf{x}_i^{1\times P^2}$ is the flattened a 1D raw vector from $X_i^{1\times P^2}$, $\mathbf{E} \in \mathbb{R}^{P^2\times D}$ is a weight matrix to convert the $ P^2$ raw pixels into a $D$-dimensional feature vector, and $\mathbf{p}^{\mathrm{enc}}_{i}$ is a positional encoding superimposed onto each feature vector to retain the %2D spatial geometry of the 
original Mel-spectrogram feature. The details to derive \cref{eq:encoder_op} %in the encoder 
is given in \cite[Appendix~C]{complementary}.

\subsubsection{Decoder}
The decoder uses $\mathbf{Z}_{\mathrm{enc}}$ to estimate the missing acoustic information in $\ddot{\mathbf{X}}$ and reconstruct a complete Mel-spectrogram $\tilde{\mathbf{X}} = \{\tilde{X}_n\in \mathbb{R}, n\in \{1,2,...,N\} \}$. For $n \in \{1, 2, \dots, N\}$, $\tilde{X}_n$ from the decoder is formulated as
\vspace{-0.3cm}

{\small 
\begin{align}
    &\tilde{X}_n = \mathcal{R}\left(\tilde{\mathbf{z}}_n\mathbf{W}_{\mathrm{pred}} + \mathbf{b}_{\mathrm{pred}}\right),\notag \\
    &\tilde{\mathbf{z}}_n \in f_{\mathrm{dec}}\left( \left[ [\mathbf{z}_{\mathrm{cls}}'; \mathbf{h}_1; \mathbf{h}_2; \dots; \mathbf{h}_{N}] + \mathbf{P}_{\mathrm{dec}} \right] \right),\notag \\
    &\mathbf{h}_i = \begin{cases} \mathbf{z}_{i} \mathbf{W}_{\mathrm{dec}}, & \text{if } i \in \mathcal{V}.
\\ \mathbf{t}_{\mathrm{mask}}, & \text{if } i \in \mathcal{U}. \end{cases} \label{eq:decoder_op}
\end{align}
}%
In \cref{eq:decoder_op}, \ $\mathcal{R(\cdot)}$ denotes the reshaping operator, $\mathbf {W}_ {\mathrm{pred}} \in \mathbb{R}^{D \times P^2}$ is responsible for mapping the $D$-dimensional hidden information back into $P^2$ distinct numerical intensities, $\mathbf {b}_ {\mathrm{pred}} \in \mathbb{R}^{1\times P^2}$ is applied to shift the predicted values to their correct physical baseline, $\tilde{\mathbf{z}}_n$ represents the $n$-th estimated feature vector with the missing acoustic fragments being successfully recovered, $\mathbf{z}_{\mathrm{cls}}'$ is the updated classification label prepended to the sequence's beginning, $\mathbf{z}_i$ is extracted from $\mathbf{Z}_{\mathrm{enc}}$ by multiplying a feature dimensional weight matrix $\mathbf{W}_{\mathrm{dec}}$, and $\mathbf{t}_{\mathrm{mask}}$ serves as a shared learnable placeholder indicating that the patch has been masked. The details to derive \cref{eq:encoder_op} in the encoder are given in \cite[Appendix~D]{complementary}.

\vspace*{-0.2cm}
\subsection{Masked Encoder based Anomaly Detection}
\label{sec: topK}
\vspace*{-0.1cm}
%During the flight of the drone, the distance between the drone and the victim constantly changes, resulting in significant fluctuations in the received volume. To eliminate the impact of these volume fluctuations and force the MAD model to focus only on the structural patterns of the sound rather than absolute loudness, we applied block by block normalization. For each initial image patch $X_n \in \mathbf{X}$, the normalized patch $\mathbf{y}_n$ is given by $\mathbf{y}_n = (X_n - \mu_n)/\sigma_n$, where $\mu_n$ and $\sigma_n$ are the mean and standard deviation of the pixel values within the $n$-th patch, respectively. Because the network was trained exclusively on normal background noise, its internal weights are strictly optimized to reconstruct only familiar environmental sounds. Consequently, when an anomalous acoustic noise (e.g., a human voice) occurs, The MAD fails to reconstruct these unlearned noise accurately, leads to significant errors between the restructured image patch $\hat{\mathbf{y}_{i}}$ and the initial image patch$\mathbf{y}_{i}$. 
To prevent large-energy background noises from masking anomaly signals, a Top$-K$ scoring strategy is used to calculate $D_{\mathrm{re}}$, i.e.,  the mean squared error of the top $K\%$ image patches (their index set is denoted as $\mathcal{K}$) with the largest reconstruction errors. Thus, the values of $D_{\mathrm{re}}$ can be calculated by
\vspace{-0.2cm}

{\small 
\begin{equation}
D_{\mathrm{re}} = \frac{1}{|\mathcal{K}|} \sum_{n \in \mathcal{K}} \|\bar{X}_n - \tilde{X}_n\|_2^2,
\label{eq:topk_score}
\end{equation}
}%

\vspace{-0.2cm}
\noindent where $\bar{X}_n = (X_n - \mu_n)/\sigma_n$ is the normalized ground-truth image patch with $\mu_n$ and $\sigma_n$ being the mean and standard deviation of the pixel values in $X_n$, respectively. $\tilde{X}_n$ is the image patch reconstructed from the MAE model. When $D_{\mathrm{re}} > D_{\mathrm{th}}$, an anomaly is detected, triggering the ring buffer and waking up the Responder stage for multi-channel localization.

\vspace*{-0.2cm}
\section{Experiments and Results}
\label{sec:experiments}
\vspace*{-0.1cm}
%In this section, comprehensive experiments are conducted to evaluate the effectiveness of V$^2$SDLS. %We first fine-tune the optimal masking ratio ($\rho$) of MAE and evaluate MAE-enabled detection accuracy (during Sentinel stage) for two typical SAR scenarios (i.e., desert and forest). %Then we demonstrate the system's localization accuracy across a UAV's flight trajectory.
%Then we simulate acoustic scenes in reality based on free space path loss and evaluate the overall performance of the system under SAR tasks.

\subsection{Experiment Settings}
\textit{Audio Dataset:} 
%To ensure the effectiveness of MAE in anomaly detection, it
MAE models are pre-trained to model a variety of \textit{Background Noise} in different SAR scenarios. \textit{Victim Sounds}, conversely, are strictly reserved for the evaluation phase to verify whether the injection of anomalous human vocalizations successfully triggers the anomaly detection condition. %To pre-train MAE models in different environments, i.e., under different background noise.
Thus, we construct two audio datasets below: 
\begin{itemize}[leftmargin=*,align=left]
\small
 \item \textit{Noise Dataset} includes two types of noise audio samples: i) the UAV ego-noise, using audio recordings ($133.3$ seconds) of a DJI drone from~\cite{MAEmodels}, which covers dynamic rotor noise and motor harmonics under various flight states, e.g., hovering, ascending, and cruising; ii) two environmental noise, including ``desert''-scenario audio recordings ($180.2$ seconds) of wind and arid environmental sounds from~\cite{audio_desert1, audio_desert2} and ``forest''-scenario audio recordings ($669.8$ seconds) of natural vegetation and bird sounds from~\cite{audio_forest}.
\item \textit{Victim Sounds} consist of audio recordings ($11,182$ seconds) of real human distress vocalizations, covering audio clips of children crying ($8,639$ seconds) and male shouts for help ($2,543$ seconds), primarily extracted from \cite{landry2020asvp}. %During the model evaluation phase, these unseen signals are mathematically superimposed onto the background noise at various SNRs. This operation precisely simulates the actual acoustic scenarios where a victim's voice is buried within the background noise.
\end{itemize}
For the feasibility of audio dataset in SAR missions~\cite {liu2018soundscape,gupta2026droneaudioset}, audio power are scaled as follows: the ``desert"-scenario audio $\sim$ $25$ dB, the ``forest"-scenario audio $\sim$ $35$ dB, the UAV ego-noise recording $\sim$ $75$ dB, and the victim sounds $\sim$ $120$ dB\footnote{We consider that the victim was in a state of distress, and the screaming sound he emitted far exceeded the normal volume. According to the shouting volume experiment of \cite {boren2013maximum}, we set the victim sound to $120$ dB.}. %Then, we think that the victim was in a state of distress, and the screaming sound he emitted far exceeded the normal volume. According to the shouting volume experiment of \cite {boren2013maximum}, we set the victim sound to 120 $dB$. %All of the above sound sources comply with free space propagation attenuation.

\textit{Pretraining of MAE models:} 
%To evaluate the proposed method, extensive model training and simulations are conducted. 
The pre-training of MAE models is %built on an Apple MacBook M5 Pro equipped with 64 GB of RAM. The software environment is 
implemented in Python 3.14 based on the PyTorch 2.10.0 deep learning framework. 
Using \textit{Noise Dataset}, multiple MAE models (with different masking ratios $\rho$) are pre-trained to capture the acoustic features using the UAV noise mixed with the background noise of desert and forest. %Each MAE model is trained in $60$ epochs, with the \textit{Background Noise} dataset divided into a 90\% training set and a 10\% validation set for cross-validation, and trained using the AdamW optimizer and mixed-precision. Finally, used the complete dataset for fixed learning of the \textit{MAE} model, this \textit{MAE} model will be used in subsequent experiments. 
For page-limit reasons, detailed training processes are given %and the MAE training hyperparameters 
in \cite{MAEmodels}.

\textit{Acoustic Scene Simulation:} 
To simulate the observed victim's sound audio, the test audio needs to be calibrated using sound attenuation along its propagation. With the UAV's altitude, LoS-dominated sound attenuation is used, % from the victim to the UAV, %i.e., the received signal power equates to the initial source level minus the total path loss (
i.e., $1/d^{\alpha}$, where $d$ is the distance from the victim to the UAV and $\alpha=\{2,2.5\}$ are the path loss factors in the desert and forest scenarios, respectively. $\alpha=2$ indicates a free space propagation and $\alpha=2.5$ accounts for the scattering of dense vegetation. %Thus, the observed energy of the victim sound is obtained by converting the SNR ratio from decibels to a linear ratio and multiplying it by the background noise energy. Next, the scaling factor is calculated by dividing the target energy by the original energy of the original victim audio. Finally, 
Given different values of $d$ and $\alpha$ in two scenarios, the victim's waveform is multiplied by the attenuation $1/d^{\alpha}$. 

\textit{Test Audio Simulation:} 
Victim sounds are added to the background noise to measure the system performance. Particularly, each MAE model and its enabled system performance are measured by a simulated test audio. The test audio has $12$ seconds and is generated by injecting a $2$-second anomaly audio (combining two randomly chosen $1$-second audios from different victim sounds) into a $10$-second background noise (the mixed UAV-ego noise and environmental noise). The injected position is random in the $12$-second test audio. 

\begin{figure}[t]
        %\captionsetup[subfigure]{justification=centering}
	\centering
    %\subfloat[Forest Scenario.]{
        \includegraphics[width=8.5cm]{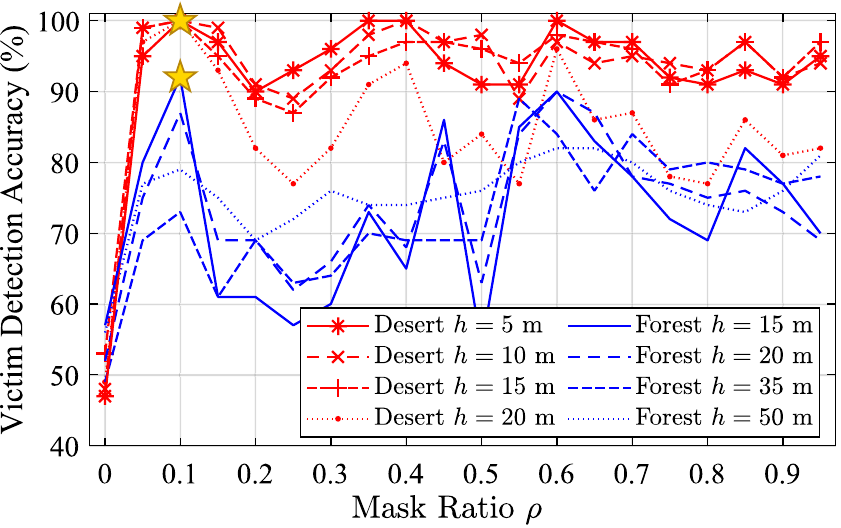} 
    \caption{Anomaly detection accuracy of MAE models versus the masking ratios $\rho$ and two scenarios: desert and forest. The yellow star indicates the largest accuracy in two scenarios. }
    \label{fig:MAE}
	 \vspace{-0.3cm}
\end{figure}

\begin{figure*}[t]
    \centering
    \includegraphics[width=17.0cm]{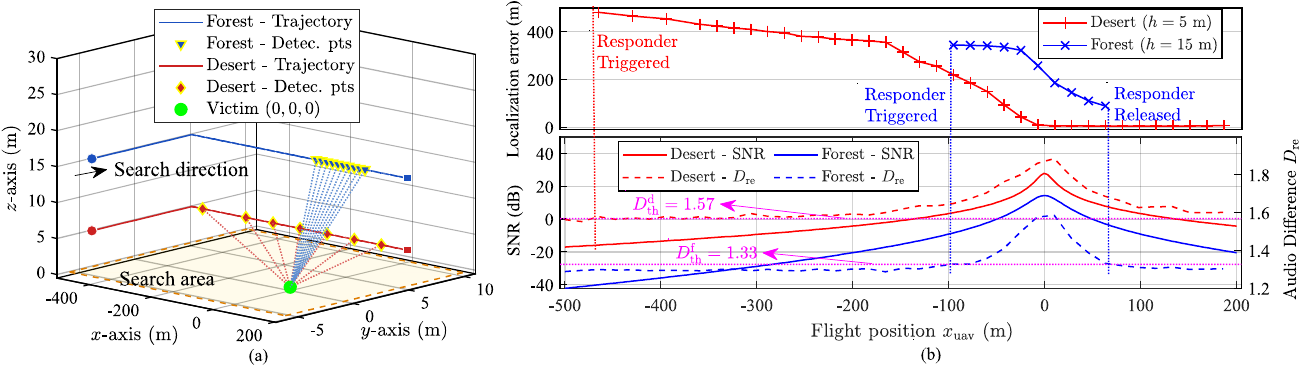}
    % \subfloat[]{
    % \includegraphics[width=5.8cm]{Images/flyover_3d.pdf}
    %     \label{fig:flyover}
    % }
    % \hfil
    % \subfloat[]{
    % \includegraphics[width=6.3cm]{Images/result_SNR_detection.pdf}
    %     \label{fig:SNR_d}
    % }
    % \hfil
    % \subfloat[]{
    % \includegraphics[width=4.1cm]{Images/result_localization.pdf}
    %     \label{fig:localization_b}
    % }
    \caption{The continuous localization results of ``Sky-Ear'' along a UAV's trajectory in SAR. Two scenarios, i.e., the desert and forest, are considered by applying $\{h = 15\text{ m},\rho=0.1 \text{(for MAE)}\}$ and $\{h=5\text{ m}, \rho=0.1\text{(for MAE)}\}$, respectively. The victim's location is set as $(0,0,0)$, the initial takeoff coordinates for the UAV are $(-500,-4,15)$ in the forest and $(-500,-4,5)$ in the desert, it takes a short flight segment along the $y$-axis from $y_{\mathrm{uav}}=-4 \sim 5\text{ m}$ on a fixed route, and cruises along the $x$-axis from $x_{\mathrm{uav}}=-500 \sim 200\text{ m}$ at $y_{\mathrm{uav}}=5\text{ m}$.}
    \label{fig:flyover}
    \vspace{-0.3cm}
\end{figure*}

\textit{System performance metrics:} The system performance is evaluated in two metrics: i) the \textit{detection Accuracy}, calculated as a MAE model’s successful detection rate averaged over evaluations of this model in $100$ testing audios; ii) the \textit{localization error}, calculated as the average Euclidean distance between the estimated position and the true victim position for each UAV's observation. Recall that an anomaly is detected when $D_{\mathrm{re}}$ (MAE's output) exceeds $D_{\mathrm{th}}$. The threshold $D_{\mathrm{th}}$ is set to 1.57 and 1.33 in desert and forest scenes. 

%\begin{figure}[t]
%    \centering
%    \includegraphics[width=\columnwidth]{Images/fig_flyover.png}
%    \caption{...}
%    \label{fig:flyover}
%\end{figure}
\vspace{-0.2cm}
\subsection{Find-tuning and Evaluation of MAE}
\label{subsec: mae_exp}
\vspace*{-0.1cm}
Based on the \textit{Victim Sounds} dataset, we first fine-tune and measure MAE models for anomaly detection. In our measurement, both the victim and the UAV are in the same horizontal position. Different UAV altitudes $h$ are considered, which inherently leads to varying received signal-to-noise (SNRs). Considering the safety constraints of practical UAV operations, we set $h=\{5,10,15,20\}$ m and $h=\{15,20,35,50\}$ m for the desert and forest scenarios, respectively. A total of $34$ MAE models are measured in terms of $17$ masking ratios ($\rho$ ranging from $0.00$ to $0.90$) and $2$ typical SAR scenarios, i.e., desert and forest acoustic environments. According to the measured results, the best $\rho$ can be observed and thus fine-tuned. \cref{fig:MAE} shows the experiment results of anomaly detection accuracy of $34$ MAE models under different UAV heights. The fluctuation occurs because of the randomly generated test audio for each measurement. 
%$\rho$is tested for different flight heights in \textit{Desert} and \textit{Forest} scenarios. 
%It is observed that, in the desert scenario, the system achieved a largest detection accuracy of 100\% when $h = 5$ m and $\rho = 0.10$. Comparatively, in forest scenario, a largest accuracy reaches approximately 93\% whem $\rho = 0.10$ and $h = 15$ m, which is slightly smaller because of the more complex propagation conditions. 
Overall, it is observed that the largest accuracy of all MAE models occurs in a relatively low masking ratio, i.e., $\rho = 0.10$. This is because a low masking rate ensures that the decoder receives sufficient contextual acoustic features to accurately identify structural deviations. %On the contrary, an excessively high masking rate will cause the network to lack necessary acoustic clues, inevitably leading to reconstruction collapse and thus reducing detection reliability. 
Comparing two scenarios, the forest scenario obtains smaller detection accuracy because of more complex propagation conditions. For each MAE model, the lowest altitude corresponds to the best accuracy because of the shorter distance. 

\vspace{-0.2cm}
\subsection{System-level Performance}
\vspace*{-0.1cm}
To evaluate the continuous localization performance of ``Sky-Ear'', %we constructed a flight simulation scene in 3D space. %The experiment assumes that the victim $P$ is stationary at the ground coordinate origin and continues to call for help (a sound pressure level of 120 dB). The UAV executes a predefined L-shaped search trajectory. 
multiple observations are conducted along a UAV's trajectory. 
%To accurately simulate real-world acoustic scenarios, we dynamically model the signal propagation based on Spherical-Spreading-Loss. 
%Specifically, in an unbounded free field, when sound waves propagate outward from a ground victim, the energy of the sound waves spreads outwards like an expanding invisible balloon (spherical wave), and the total emitted sound power is evenly distributed on the expanding sphere. Therefore, the decrease in sound intensity is inversely proportional to the square of spatial distance ($P\propto 1/d^2$).
\cref{fig:flyover} shows the continuous localization results of ``Sky-Ear'' in two SAR scenarios. \cref{fig:flyover} (a) illustrates the UAV's flight trajectory, and \cref{fig:flyover} (b) records the dynamic evaluation of the system along the trajectory. Overall, experimental results verify the effectiveness of our system. It is observed that the UAV initially approaches from a distance; thus, both the received SNR and $D_\mathrm{re}$ remain extremely low. With the approach of the drone to the victim, both the received SNR and $D_\mathrm{re}$ increase. We can see a sequence of blank spaces ``Localization Error", indicating that the responder stage is not triggered and the system is still in the sentinel stage during long-distance flight. Once the responder stage is triggered for localization, we can observe that, as the drone approaches the victim, the positioning error drops sharply. 
%\begin{itemize}
%\item \textbf{Sentinel stage (orange line):} 
%As depicted by the orange trajectory on the left, the UAV initially approaches from a distance. Correspondingly, in the subplots (``SNR \& Detection''), both the received SNR (solid green line) and the MAE reconstruction error (dashed $orange$ line) remain extremely low. The UAV did not receive the sound signal from victim P, and the MAE did not trigger any anomaly signals. As the UAV approaches the victim,  both the SNR and the MAE error to spike simultaneously and reach their absolute peaks exactly at $X = 0$.
%This indicates that during long-distance approach, the system remains in a lightweight cruising mode, as low SNR ratio renders localization useless. Only when the UAV approaches an uneducated distress signal, its reconstruction ability will be compromised, and the reconstruction errors $D$ will be greater than the threshold $D_\mathrm{th}$ we set, triggering the Respondent localization phase.
%\item \textbf{Responder stage (blue line):} Once the MAE error breaches the detection threshold the peak, the system immediately switches to the Responder stage, as shown in the subgraph in the bottom left corner(blue trajectory), initiating high-frequency spatial acoustic sampling (dashed projection lines). Simultaneously, as shown in the lower right subplots (``Localization Error''), as the drone approaches the victim, the positioning error drops sharply, and the final positioning error is less than $5m$.
%\end{itemize}
Comparing the two scenarios, %the desert environment exhibits a sharply defined peak in both SNR and detection scores, accompanied by a rapid convergence in localization error. %Conversely, 
the forest scenario presents a noticeably blunter peak and a delayed localization convergence. This performance discrepancy fundamentally stems from the constraints of the flight altitude and canopy attenuation. In the forest, the higher altitude causes the relative propagation distance $d$ to change more gradually as the UAV flies over the victim; thus, the received SNR rises more smoothly. Furthermore, the continuous localization requires a longer flight path to accumulate sufficient spatial disparity, which directly explains the delayed convergence in the localization error. 

\vspace{-0.2cm}
\section{Conclusion}
\label{sec:conclusion}
\vspace{-0.1cm}
%The framework centers on a dual-stage Sentinel-Responder architecture bridged by a retroactive ring buffer. By utilizing a scenario-adaptive MAE design to maintain low idle power and leveraging historical data to eliminate wake-up information loss, the system achieves significant energy savings while maintaining robust spatial localization.
We proposed ``Sky-Ear'', an energy-efficient acoustic-based victim-detection-and-localization system for UAV-enabled SAR missions. The system functions as a circular microphone array, two-stage (Sentinel and Responder) audio processing, and continuous localization. The Sentinel stage uses an MAE on Mel-spectrograms to learn background and UAV noise and detect victim sounds as anomalies, while the Responder stage performs detailed analysis only when needed for computation. A continuous localization method is designed by optimizing detected directions from multiple observations. Using custom audio datasets with environmental, propeller, and victim sounds, simulation results show that ``Sky-Ear'' achieves accurate victim detection and improved localization by multiple observations along the UAV trajectory, demonstrating its effectiveness. In future work, we plan to expand our pre-training datasets to real-world UAV deployments, e.g., capturing complex aeroacoustic dynamics, such as severe rotor turbulence and structural vibrations

% We proposed Sky-Ear, an energy-efficient acoustic system for UAV-enabled SAR missions. It utilizes a circular microphone array and two-stage audio processing: an MAE-based Sentinel stage for low-power anomaly (victim sound) detection, and an on-demand Responder stage for precise localization. By optimizing detected directions across continuous flight observations, simulations on custom datasets verify that Sky-Ear achieves highly accurate localization while significantly minimizing computational overhead.

% In future work, we will focus on transitioning Sky-Ear from simulated environments to real-world UAV deployments. To mitigate complex aeroacoustic dynamics, such as severe rotor turbulence and structural vibrations, we plan to implement hardware-level isolation and significantly expand our pre-training datasets. Furthermore, we will prioritize advancing our localization framework. By exploring TDoA estimation techniques, including phase-weighted GCC-PHAT and adaptive filtering, we aim to enhance the stability of spatial direction estimation under dynamic noise conditions. Finally, these deployments will enable comprehensive quantitative evaluations of the system's energy consumption, computational overhead, and response latency.
\vspace{-0.2cm}

\bibliographystyle{IEEEtran}
\bibliography{references}

\end{document}